\RequirePackage{fix-cm}
\documentclass[smallextended]{svjour3}       % onecolumn (second format)
\smartqed  % flush right qed marks, e.g. at end of proof
\usepackage{graphicx}

\usepackage{color}
\usepackage{amsfonts,amsmath,algorithm,graphicx,subfigure}
\usepackage{booktabs} 
\usepackage[noend]{algorithmic}
\usepackage[numbers]{natbib}
\newcommand{\la}{\ell}

\begin{document}

\title{A Faster Algorithm for Cuckoo Insertion and Bipartite Matching in Large Graphs \footnote{An extended abstract of this work appeared in the \emph{Proceedings of the 21st Annual European Symposium on Algorithms(ESA '13)\cite{inc:khosla13}.}
}}

%\thanks{Grants or other notes
%about the article that should go on the front page should be
%placed here. General acknowledgments should be placed at the end of the article.}

%\subtitle{Do you have a subtitle?\\ If so, write it here}

%\titlerunning{Short form of title}        % if too long for running head

 \author{Megha Khosla       \and
         Avishek Anand %etc.
 }

% %\authorrunning{Short form of author list} % if too long for running head

 \institute{Megha Khosla \at
               %first address \\
              % Tel.: +123-45-678910\\
              % Fax: +123-45-678910\\
                      %  \\
         \emph{L3S Research Center, Leibniz University, Hannover}\\
          \email{khosla@l3s.de}   
                  %  if needed
           \and
           Avishek Anand \at
           \emph{L3S Research Center, Leibniz University, Hannover}\\
                 \email{anand@l3s.de}
 }

\date{Received: date / Accepted: date}
% The correct dates will be entered by the editor

\newcommand{\eps}{\varepsilon}
\newcommand{\al}{\alpha}
\newcommand{\be}{\beta}
\newcommand{\Nat}{\mathbb{N}}
\newcommand{\prob}{\mathbb{P}}
\newcommand{\ex}{\mathbb{E}}
\newcommand{\Pa}{{\mathcal P}}
\newcommand{\U}{{\mathcal U}}
\newcommand{\B}{{\mathcal B}}
\newcommand{\es}{{\mathcal S}}
\newcommand{\E}{{\mathbb E}}
\newcommand{\W}{{\mathcal W}}
\newcommand{\F}{{\mathcal F}}
\newcommand{\M}{{\mathcal M}}
\newcommand{\C}{{\textsc C}}
\newcommand{\D}{{\mathcal D}}
\newcommand{\X}{{\mathcal X}}
\newcommand{\Q}{{\mathcal Q}}
\newcommand{\A}{{\mathcal A}}
\newcommand{\La}{{\mathcal L}}
\newcommand{\T}{{\mathcal T}}
\newcommand{\G}{{\mathcal{G}}}
\newcommand{\ga}{\gamma}
\newcommand{\Bin}{\mathrm{Bin}}
\newcommand{\Aux}{{\mathcal G}}
\newcommand{\Li}{{\mathcal L}}
\newcommand{\N}{{\mathcal N}}
\newcommand{\I}{{\mathcal I}}
\newcommand{\It}{\tilde{\mathcal I}}
\newcommand{\Is}{\tilde{I}}
\newcommand{\Nt}{\mathcal N}
\newcommand{\nt}{\tilde{N}}
\newcommand{\Ub}{\overline{U}}
\newcommand{\Hi}{\mathcal H}
\newcommand{\va}{\overline{a}}
\newcommand{\Exx}{\mathrm{Ex}}
\newcommand{\Pl}{\mathsf{P}}
\newcommand{\Nl}{\mathsf{N}}
\newcommand{\Cl}{\mathsf{C}}
\newcommand{\Bl}{\mathsf{B}}

\newcommand{\Y}{{\mathcal{Y}}}
\newcommand{\BP}{{\mathbb{P}}}

\newcommand{\MN}{{\mathcal{N}}}
\newcommand{\CN}{{\mathcal{N}}}
\newcommand{\CS}{{\mathcal{S}}}
\newcommand{\CE}{{\mathcal{E}}}
\newcommand{\CP}{{\mathcal{P}}}
\newcommand{\CH}{{\mathcal{H}}}
\newcommand{\CT}{{\mathcal{T}}}
\newcommand{\Net}{{\mathtt{Net}}}
\newcommand{\Ber}{{\mathtt{Be}}}
\newcommand{\Be}{{\mathtt{Be}}}
\newcommand{\Ser}{{\mathtt{Ser}}}
\newcommand{\Par}{{\mathtt{Par}}}
\newcommand{\sh}{{\mathtt{sh}}}
\newcommand{\Set}{{\mathsf{Set}}}

\renewcommand{\Pr}[1]{\mathbb{P}\left(#1\right)}

\newcommand{\keyword}[1]{\textbf{#1}~}

\newcommand{\Pcl}{\widetilde{H}_{n,p,k}}

\maketitle
\abstract{Hash tables are ubiquitous in computer science for efficient access to large datasets. However, there is always a need for approaches that offer compact memory utilisation without substantial degradation of lookup performance. Cuckoo hashing is an efficient technique of creating hash tables with high space utilisation and offer a guaranteed constant access time. We are given $n$ locations and $m$ items. Each item has to be placed in one of the $k\ge2$ locations chosen by $k$ random hash functions. By allowing more than one choice for a single item, cuckoo hashing resembles multiple choice allocations schemes. In addition it supports dynamically changing the location of an item among its possible locations. We propose and analyse an insertion algorithm for cuckoo hashing that runs in \emph{linear time} with high probability and in expectation. Previous work on total allocation time has analysed breadth first search, and it was shown  to be linear only in \emph{expectation}. Our algorithm finds an assignment (with probability 1) whenever it exists. In contrast, the other known insertion method, known as \emph{random walk insertion}, may run indefinitely even for a solvable instance.
We also present experimental results comparing the performance of our algorithm with the random walk method, also for the case when each location can hold more than one item. 

As a corollary we obtain a linear time algorithm (with high probability and in expectation) for finding perfect matchings in a special class of sparse random bipartite graphs. We support this by performing experiments on a real world large dataset for finding maximum matchings in general large bipartite graphs. We report an order of magnitude improvement in the running time as compared to the \emph{Hopkraft-Karp} matching algorithm.

} 

\section{Introduction}
In computer science, a hash table~\cite{bookC:09} is a data structure that maps items (keys) to locations (values) using a hash function. More precisely, given a universe $U$ of items and a hash table $H$ of size $n \in \mathbb{N}$, a hash function
$h : U \rightarrow \{1,2,\ldots, n\}$ maps the items from $U$ to the $n$ positions on the table. Ideally, the hash function should assign to each possible item to a unique location, but this objective is rarely achievable in practice. Two or more items could be mapped to the same location resulting in a collision. In this work we deal with a collision resolution technique known as \emph{cuckoo hashing}.  Cuckoo hashing was first proposed by Pagh and Rodler in \cite{inp:pr01}. We are interested in a generalization of the original idea (see~\cite{inc:fpss03}) where we are given a table with $n$ locations, and we assume each location can hold a single item. Each item chooses randomly $k\ge 3$ locations (using $k$ random hash functions) and has to be placed in one of them.  Formally speaking we are given $k\ge3$ hash functions $h_1, . . . , h_k$
that each maps an element $x \in U$ to a position in the table $H$. 
Moreover we assume that $h_1, . . . , h_k$ are truly independent and random hash functions. We refer the reader to ~\cite{ms:08,d:09}  (and references therein) for justification of this idealized assumption. 
 Other variations of cuckoo hashing are considered in for example \cite{inp:ans:09,art:kmw09}.

Cuckoo hashing resembles multiple choice allocations schemes in the sense that it allows more than one choice for a single item. In addition it supports dynamically changing the location of an item among its possible locations during insertion. The insertion procedure in cuckoo hashing goes as follows. Assume that $p$ items have been inserted, each of them having made their $k$ random choices on the hash table, and we are about to insert the $p+1$st item. This item selects its $k$ random locations from the hash table and is assigned to one of them. But this location might already be occupied by a previously inserted item. In that case, the previous item is evicted or ``kicked out'' and is assigned to one of the other $k-1$ selected locations. In turn, this position might be occupied by another item, which is kicked out and goes to one of the remaining $k-1$ chosen locations. This process may be repeated indefinitely or until a free loction is found. 

We model cuckoo hashing by a directed graph $G=(V,E)$  such that the set of vertices $V=\{v_1,v_2,\ldots,v_n\}$ corresponds to locations on the hash table. We say a vertex is \emph{occupied} if there is an item assigned to the corresponding location, otherwise it is \emph{free}. Let $\mathcal{I}$ be the set of $m$ items. We represent each item $x\in\mathcal{I} $ as a tuple of its $k$ chosen vertices (locations), for example, $x=(v_{x_1},v_{x_2},\ldots , v_{x_k})$.
A directed edge $e=(v_i,v_j) \in E$ if and only if there exists an item $y\in \mathcal{I} $ so that the following two conditions hold, (i) $v_i,v_j \in y$, and (ii) $v_i$ is occupied by $y$. Note that a vertex with outdegree $0$ is a free vertex. We denote the set of free vertices by $F$ and the distance of any vertex $v$ from some vertex in $F$ by $d(v,F)$. Since $G$ represents an allocation we call $G$ an \emph{allocation graph}. 

%AWe then repeat the procedure for the replaced item. We say an \emph{assignment} is made whenever an item is placed on a chosen vertex.We repeat the above procedure for the kicked out item. 

Now assume that in the cuckoo insertion procedure, at some instance an item $z$ arrives such that all its $k$ choices are occupied. Let $v_j\in z$ be the vertex chosen to place item $z$. The following are the main observations.

 \begin{enumerate}
\item The necessary condition for item $z$ to be successfully inserted at $v_{z_j}$ is the existence of a path from $v_{z_j}$ to $F$. This condition remains satisfied as long as some allocation is possible.
\item The procedure will stop in the minimum number of steps if for all $v_{z_i}\in z$  the distance $d(v_{z_j},F)\le d(v_{z_i},F)$. \end{enumerate}

With respect to our first observation, a natural question to ponder would be the following. We are given a set of $m$ items and $n$ locations such that each item picks $k\ge 3$ locations at random. Is it possible to place each of the items into one of their chosen locations such that each location holds at most one item? From~ \cite{inp:l12,fp12,fm12} we know that there exists a critical size $c_k^* n$ such that if $m<c_k^* n$ then such an allocation is possible with high probability, otherwise this is not the case.

\begin{theorem}
\label{thm:mainO}
For integers~$k\geq 3$ let~$\xi^\ast$ be the unique solution of the equation
\begin{equation}\label{eq:kxi}
 k = {\xi(1-e^{-\xi})\over 1-e^{-\xi} -\xi e^{-\xi} }.
\end{equation}
Let~$c_{k}^\ast= \frac{\xi^\ast}{k(1-e^{-\xi^{*}})^{k-1}}$. Then
\begin{equation}
\label{eq:phnmkl}
	\Pr{\text{allocation of $m=\lfloor cn\rfloor$ items to $n$ locations is possible}}
	\stackrel{(n \to \infty)}{=} 
	\begin{cases}
		0, & \text{ if } c > c_{k}^\ast \\
		1, & \text{ if } c < c_{k}^\ast
	\end{cases}.
\end{equation}
\end{theorem}

The proof of the above theorem is non-constructive, i.e., it does not give us an algorithm to find such an allocation. In this work we deal with the algorithmic issues and propose an algorithm which takes linear time with high probability and in expectation to find the optimal allocation.

Our second observation suggests that the insertion time in the cuckoo hashing depends on the selection of the location, which we make for each assignment, from among the $k$ possible locations. One can in principle use breadth first search to  always make assignments over the shortest path (in the allocation graph). But this method is inefficient and expensive to perform for each item. One can also select uniformly at random a location from the available locations. This resembles a random walk on the locations of the table and is called the \emph{random walk insertion}. In~\cite{art:fps10,art:fmm11} the authors analyzed the random walk insertion method and gave a polylogarithmic bound (with high probability) on the maximum insertion time, i.e., the maximum time it can take to insert a single item.

\subsection{More on Related Work}
The allocation problem in cuckoo hashing can also be phrased in terms of \emph{orientation of graphs} or more generally \emph{orientations of $k$-uniform hypergraphs}. The $n$ locations are represented as vertices and each of the $m$ items form an edge with its $k$-vertices representing the $k$ random choices of the item.
In fact, this is a random (multi)hypergraph $H^*_{n,m,k}$ (or random (multi)graph $G^*_{n,m}$ for $k=2$) with $n$ vertices and $m$ edges where each edge is drawn uniformly at random ( with replacement) from the set of all $k$-multisubsets of the vertex set. An $\la$-orientation of a graph then amounts to a mapping of each edge to one of its vertices such that no vertex receives more than $\la$ edges. $\la$ is also called the maximum load capacity. In our algorithm, we focus on $\la=1$. Here, we give an overview of existing work for general $\la$ for completeness.

For the case $k = 2$, several allocation algorithms and their analysis are closely connected to the cores of the associated
graph. The $\la$ core of a graph is the maximum vertex induced subgraph with minimum degree at least $\la$. As another application, the above described problem can also be seen as a load balancing problem with locations representing the \emph{machines} and the items representing the \emph{jobs}. To this extent \citet{ina:cs01} gave a linear time algorithm achieving maximum load $O(m/n)$ based on computation of all cores. The main idea was to repeatedly choose a
vertex $v$ with minimum degree and remove it from the graph, and assigning all its incident edges (items) to vertex (location) $v$. \citet{1283433}  used a variation of the above approach and gave a linear time algorithm for computing an optimal allocation (asymptotically almost surely). Their algorithm first guesses the optimal load among the two likely values values ($\lceil m/n\rceil $ or $\lceil m/n\rceil +1$). The procedure starts with a load value say $\la=\lceil m/n\rceil $. Each time a vertex with degree at most $\la$ and its incident edges are assigned to $v$. The above rule, also called the \emph{mindegree rule}, first reduces the graph to its $\la+1$ core. Next, some edge $(u,v)$ is picked according to some priority rule and assigned to one of its vertices. Again the mindegree rule is applied with respect to some conditions. In case the algorithm fails it is repeated after incrementing the load value.

\citet{1283432} used a different approach in dealing with the vertices with degree greater than the maximum load.  Their algorithm, called the \emph{excess degree reduction} (EDR) approach, always chooses a vertex with minimum degree, $d$. If $d<\la$ then this vertex is assigned all its incident edges and is removed from the graph. In case $d>2\la$ the algorithm fails. Otherwise, EDR replaces $d-\la$ paths of the form
$(u, v, w)$ by bypass edges $(u, w)$ and then orients all remaining edges ($\le \la$ ) incident to $v$
towards $v$.

\noindent Optimal allocations can also
be computed in polynomial time using maximum 
flow computations and with high probability achieve a maximum load of 
$\lceil m/n\rceil$ or $\lceil m/n\rceil +1$~\cite{sek99}. 

Recently \citet{art:aumuller16} analyzed our algorithm in their special framework of an easily computable hash class.

\paragraph{Notations.} Throughout the paper we use the following notations. We denote the set of integers $\{1,2,\ldots, n\}$ by $[n]$. Let $V=\{v_1,v_2,\ldots, v_n\}$ be the set of $n$ vertices representing the locations of the hash table. For an allocation graph $G=(V,E)$ and any two vertices $v_i,v_j\in V$, the shortest distance between $v_i$ and $v_j$ is denoted by $d(v_i,v_j).$  We denote the set of free vertices by $F$. We denote the shortest distance of a vertex $v_i\in V$ to any set of vertices say $S$ by $d(v_i,S)$ which is defined as \[d(v_i,S): = \min_{v_j\in S} d(v_i,v_j).\] We use $R$ to denote the set of vertices furthest from $F$, i.e.,
\[   R:= \{ v_i\in V |d(v_i, F ) \ge \max_{ v_j \in V} d(v_j,F)\}.
\]

For some integer $t\in[n]$ and the subset of vertex set $V' \subset V$ let ~$N_{t}(v_i)$ and $N_{t}(V')$ denote the set of vertices at distance at most $t$ from the vertex $v_i \in V$ and the set $V'$. Mathematically,
\[ N_{t}(v_i) := \{ v_j \in V~| ~d(v_i,v_j) \le t  \} \]
and 
\[ N_{t}(V') := \{ v_j \in V~| ~d(v_i,V') \le t  \} .\]
\subsection{Our Contribution}
Our aim here is to minimize the total insertion time in cuckoo hashing, thereby minimizing the total time required to construct the hash table. We propose a deterministic strategy of how to select a vertex for placing an item when all its choices are occupied.
We assign to each vertex $v_i\in V$ an integer label, $L(v_i)$. Initially all vertices have $0$ as their labels. Note that at this stage, for all  $j\in [n]$, $L(v_j) =d(v_j,F)$, i.e., the labels of all vertices represent their shortest distances from $F$. When an item $x$ appears, it chooses the vertex with the least label from among its $k$ choices. If the vertex is free, the item is placed on it. Otherwise, the previous item is kicked out. The label of the location is then updated and set to one more than the minimum label of the remaining $k-1$ choices of the item $x$. The kicked out item chooses the location with minimum label from its $k$ choices and the above procedure is repeated till an empty location is found. 
Note that to maintain the labels of the vertices as their shortest distances from $F$ we would require to update labels of the neighbors of the affected vertex and the labels of their neighbors and so on. This corresponds to performing a breadth first search (bfs) starting from the affected vertex. We avoid the bfs and perform only local updates. Therefore, we also call our method as \emph{local search allocation}. 

Previous work~\cite{inc:fpss03} on total allocation time has analysed breadth first search, and it was shown  to be linear only in \emph{expectation}.  The local search allocation method requires linear time with probability $1-o(1)$  and in expectation to find an allocation. We now state our main result.

\begin{theorem}\label{thm:main}
Let $k\ge 3$. For any fixed $\eps >0$, set $m=(1-\eps) c_k^* n$.  Assume that each of the $m$ items chooses $k$ random locations (using $k$ random hash functions) from a table with $n$ locations. With probability $1-O(n^{-1})$ , LSA finds an allocation of these items (such that no location holds more than one item) in time O(n). Moreover the expected running time of LSA is always O(n), regardless whether there exists an allocation or not.  
\end{theorem}
We prove the above theorem in two steps. First we show that the algorithm is correct and finds an allocation in polynomial time. To this end we prove that, at any instance, label of a vertex is at most its distance from the set of free vertices. Therefore, no vertex can have a label greater than $n$. This would imply that the algorithm could not run indefinitely and would stop after making at most $n$ changes at each location. 
We then show that the local search insertion method will find an allocation in a time proportional to the sum of distances of the $n$ vertices from $F$ (in the resulting allocation graph). 
We then complete the proof by showing that  $(i)$ if for some $\varepsilon >0$, $m=(1-\varepsilon)c^*_k$ items are placed in $n$ locations using $k$ random hash functions for each item then the corresponding allocation graph has two special structural properties with probability $1-o(1)$, and $(ii)$ if  the allocation graph has these two properties, then the sum of distances of its vertices from $F$ is linear in $n$. In the next section we give a formal description of our algorithm and its analysis.

\section{Local Search Insertion and its Analysis}
Assume that we are given items in an online fashion, i.e., each item chooses its $k$ random locations whenever it appears. Moreover, items appear in an arbitrary order.  The insertion using local search method goes as follows. For each vertex $v\in V$ we maintain a label. Initially each vertex is assigned a label $0$. To assign an item $x$ at time $t$ we select one of its chosen vertices $v$ such that its label is minimum and assign $x$ to $v$. We assign a new label to $v$ which is one more than the minimum label of the remaining $k-1$ choices of $x$. However,  $v$ might have already been occupied by a previously assigned item $i'$. In that case we kick out $y$ and repeat the above procedure. Let $\mathbf{L}= \{ L(v_1), \ldots, L(v_n)\}$  and $\mathbf{T}= \{ T(v_1), \ldots, T(v_n)\}$ where $L(v_i)$ denotes the label of vertex $v_i$ and $T(v_i)$ denotes the item assigned to vertex $v_i$.  We initialize $\mathbf{L}$ with all $0$s , i.e., all vertices are free. We then use Algorithm~\ref{algo:orientEdge} to assign an arbitrary item when it appears. 
\begin{algorithm}[h!]
\caption{AssignItem ($x, \mathbf{L},\mathbf{T}$)}
\label{algo:orientEdge}
\begin{algorithmic}[1]
\STATE Choose a vertex $v$ among the $k$ choices of $x$ with minimum label $L(v)$.
\IF{$(L(v)>=n-1)$ }
\STATE $\mathbf{EXIT}$  ~~~~~~~~~~~~~~~~~~~~~~~ $\rhd${\textbf{Allocation does not exist}}
\ELSE
\STATE $L(v) \leftarrow 1+ \min{(L(u)| u \neq v \text{~and $u \in x$})}$
\IF{$(T(v)\neq \emptyset )$}
\STATE $y\leftarrow T(v)$~~~~~~~~~~~~~~~~~~ $\rhd${\textbf{Move that replaces an item}}
\STATE $T(v) \leftarrow x$
\STATE $\mathbf{CALL}$ {AssignItem($y, \mathbf{L},\mathbf{T}$)}
\ELSE  
\STATE $T(v) \leftarrow x$ ~~~~~~~~~~~~~~~~~~ $\rhd${\textbf{Move that places an item}}
\ENDIF
\ENDIF
\end{algorithmic}
\end{algorithm}
In the next subsection we first prove the correctness of the algorithm, i.e, it finds an allocation in a finite number of steps whenever an allocation exists. We show that the algorithm takes a maximum of $O(n^2)$ time before it obtains a mapping for each item. We then proceed to give a stronger bound on the running time.

\subsection{Labels and the Shortest Distances } \label{sec:proof}

We need some additional notation. In what follows  a \emph{move} denotes either placing an item in a free vertex or replacing a previously allocated item.
Let $M$ be the total number of moves performed by the algorithm. For $p\in [M]$ we use $L_p(v)$ to denote the label of vertex $v$ at the end of the $p$th move. Similarly we use $F_p$ to denote the set of free vertices at the end of $p$th move. The corresponding allocation graph is denoted as $G_p=(V,E_p)$.
We need the following proposition.

\begin{proposition}\label{prop:lev}
For all $p\in[M]$ and all $v\in V$, the shortest distance of $v$ to $F_p$ is at least the label of $v$, i.e., $d(v,F_p)\ge L_p(v)$.
\end{proposition}
\begin{proof}
We first note that the label of a free vertex always remain $0$, i.e.,
\begin{align}\label{obs1} \forall p\in [M], \forall w\in F_p, ~~~~~L_p(w)=0. \end{align}
We will now show that throughout the algorithm the label of a vertex is at most one more than the label of any of its immediate neighbors (neighbors at distance $1$). More precisely,
\begin{align}\label{obs2}
\forall p\in [M], \forall (u,v)\in E_p, ~~~~~ L_p(u)\le L_p(v)+1. \end{align}
We prove \eqref{obs2} by induction on the number of moves performed by the algorithm. Initially when no item has appeared all vertices have $0$ as their labels. When the first item is assigned, i.e., there is a single vertex say $u$ such that $L_1(u)=1$. Clearly, \eqref{obs2} holds after the first move. Assume that \eqref{obs2} holds after $p$ moves.

For the ($p+1$)th move let $w\in V$ be some vertex which is assigned an item $x$. Consider an edge $(u,v)\in E_p$ such that $u\neq w$ and $v\neq w$. Note that the labels of all vertices $v \in V\setminus w$ remain unchanged in the ($p+1$)th move. Therefore by induction hypothesis, \eqref{obs2} is true for all edges which does not contain $w$. By Step $2$ of Algorithm~\ref{algo:orientEdge} the new label of $w$ is one more than the minimum of the labels of its $k-1$ neighbors, i.e, \[L_{p+1}(w) =  \min_{w'\in x\setminus w} L_{p+1}(w') +1.\] Therefore \eqref{obs2} holds for all edges originating from $w$. Now consider a vertex $u\in V$ such that $(u,w)\in E_p$. Now by induction hypothesis we have $L_{p+1}(u)=L_p(u)\le L_p(w) +1.$ Note that the vertex $w$ was chosen because it had the minimum label among the $k$ possible choices for the item $x$, i.e., 
\[ L_p(w) \le \min_{w'\in x} L_{p} (w') = \min_{w'\in x\setminus w} L_{p+1} (w') < L_{p+1}(w).\] We therefore obtain
$ L_{p+1}(u)\le L_p(w) +1 < L_{p+1}(w) +1,$
thereby completing the induction step.
We can now combine \eqref{obs1} and \eqref{obs2} to obtain the desired result. To see this, consider a vertex $v$ at distance $s< n$ to a free vertex $f\in F_p$ such that $s$ is also the shortest distance from $v$ to $F_p$.  By iteratively applying \eqref{obs2} we obtain $L_p(v) \le s+ L_p(f) =d(v,F_p)$, which completes the proof.
\end{proof}

We know that whenever the algorithm visits a vertex, it increases its label by at least 1. Trivially the maximum distance of a vertex from a free vertex is $n-1$ (if an allocation exists), and so is the maximum label. Therefore the algorithm will stop in at most $n(n-1)$ steps, i.e., after visiting each vertex at most $n-1$ times, which implies that the algorithm is correct and finds an allocation in $O(n^2)$ time. In the following we show that the total running time is proportional to the sum of labels of the $n$ vertices.
\begin{lemma}\label{lem:graph}
Let $\mathbf{L^*}$ be the array of labels of the vertices after all items have been allocated using Algorithm~\ref{algo:orientEdge}. Then the total time required to find an allocation is $O(\sum_{v\in V} L^*(v))$.
\end{lemma}
\begin{proof}
Now each invocation of Algorithm \ref{algo:orientEdge} increases the label of the chosen vertex by at least 1. Therefore, if a vertex has a label $\ell$ at the end of the algorithm then it has been selected (for any move during the allocation process) at most $\ell$ times. Now the given number of items can be allocated in a time proportional to the number of steps required to obtain the array $\mathbf{L}^*$ (when the initial set consisted of all zeros)  and hence is $O(\sum_{v\in V} L^*(v))$. 
\end{proof}
For notational convenience let $F:= F_M$ and $G:= G_M$ denote the set of free vertices and the allocation graph (respectively) at the end of the algorithm.
By Proposition~\ref{prop:lev}  we know that for each $v\in V$, $L^*(v) \le d(v,F)$. Moreover, by Step $2$ of Algorithm~\ref{algo:orientEdge} the maximum value of a label is $n$.
Thus the total sum of labels of all vertices is bounded as follows.\[\sum_{v_i\in V} L^*(v_i)) \le \min\left(\sum_{v_i\in V}d(v,F), n^2\right).\]
So our aim now is to bound the shortest distances such that the sum of these is linear in the size of $G$. We accomplish this in the following section.

\subsection{Bounding the Distances }
To compute the desired sum, i.e., $\sum_{v_i\in V}d(v,F)$, we study the structure of the allocation graph. We use the following lemma from \cite{art:fps10} (see Corollary 2.3 in \cite{art:fps10}) which states that, with probability $1-o(1)$, a fraction of the vertices in the allocation graph are at a constant distance to the set of free vertices, $F$. This would imply that the contribution for the above sum made by these vertices is $O(n)$. 

\begin{lemma}\label{lem:dist}
For any fixed $\varepsilon>0$, let $m=(1-\varepsilon) c_k^* n$ items are assigned to $n$ locations using $k$ random choices for each locations. Then the corresponding allocation graph $G=(V,E)$ satisfies the following with probability $1-O(1/n)$: for every $\alpha >0$ there exist $C=C(\alpha, \varepsilon) >0$ and a set $S \subseteq V$ of size at least $(1-\alpha)n$ such that every vertex $v\in S$ satisfies $d(v,F)\le C$. \end{lemma}

With respect to an allocation graph recall that we denote the set of vertices furthest from $F$ by $R$. Also for an integer $s$, $N_s(R)$ denotes the set of vertices at distance at most $s$ from $R$. The next lemma states that the neighborhood of $R$ expands suitably with high probability. We remark that the estimate, for expansion factor, presented here is not the best possible but nevertheless suffices for our analysis. 
\begin{lemma}\label{lem:expan}
For any fixed $\varepsilon>0$, let $m=(1-\varepsilon) c_k^* n$ items are assigned to $n$ locations using $k$ random choices for each item and $G=(V,E)$ be the corresponding allocation graph. Let .
Then for $\alpha < (e^{k}(k-2))^{-1\over{k-2}}(k-1)^{-1}$ and $0<\gamma < k-2$ and every integer $s$ such that $n^{1/2} < | N_{s}(R)|\le \alpha n$, $G$  satisfies the following with probability $1-e^{-O(n^{0.5})}$. For the case $\log n < | N_{s}(R)|\le n^{1/2}$, the following holds with probability $1-n^{-\zeta}$ for some $\zeta>0$.

\[
	|N_{s}(R) | >  \left( 1+\gamma \right)  | N_{s-1}(R)| . \]
	
\end{lemma}

As already mentioned we can model the allocation problem in cuckoo hashing as a hypergraph. Each location can be viewed as a vertex and each item as an edge. The $k$ vertices of each edge represent its $k$-random choices. In fact, this is a random hypergraph  with $n$ vertices and $m$ edges where each edge is drawn uniformly at random (with replacement) from the set of all $k$-multisubsets of the vertex set.
Therefore, a proper allocation of items is possible if and only if~the corresponding hypergraph is~\emph{$1$-orientable}, i.e., if there is an assignment of each edge $e\in E$ to one of its vertices $v\in e$ such that each vertex is assigned at most one edge. We denote a random (multi)hypergraph with $n$ vertices and $m$ edges by $H_{n,m,k}$. We will show that Lemma \ref{lem:expan} follows directly from the following expansion properties of $H_{n,m,k}$.

\begin{lemma}\label{lem:expanhyper}
Let $m<c_k^*n$ and $\alpha < (e^{k}(k-2))^{-1\over{k-2}}(k-1)^{-1}$ and $0<\gamma < k-2$. Then for every integer $s$ such that $n^{1/2} \le s \le \alpha n$, the number of vertices spanned by any set of edges of size $s$ in $H_{n,m,k}$ is greater than $ \left( 1 +\gamma \right)s$ with probability $1-e^{-O(n^{0.5})}$. For $\log n\le s < n^{1/2}$, the above holds with probability $1-n^{-\zeta}$ for some $\zeta >0$. 
\end{lemma}
\begin{proof}
Recall that each edge in $H_{n,m,k}$ is a multiset of size $k$. Therefore, the probability that an edge of  $H_{n,m,k}$ is contained completely in a subset of size $t$ of the vertex set  is given by ${t^k\over n^k}$.
Thus the expected number of sets of edges of size $s$ that span at most $t$ vertices is at most
${m\choose s}{n\choose t} \left({t^k\over n^k}\right)^{s}.$
% Note that  by the following approximation for factorials for positive integer $a$
% \[\left({a\over e}\right)^a  \sqrt{2\pi a}\le a! \le \left({a\over e}\right)^a e \sqrt{a} ,
% \]
% we obtain for $0<b<a$, ${a\choose b} < \left({ae\over b}\right)^b$.
% \begin{align*}
% {a\choose b} = {a! \over b! (a-b)! }\le &{ \left({a\over e}\right)^a e \sqrt{a}\over \left({b\over e}\right)^b \left({a-b\over e}\right)^{a-b}\sqrt{2\pi b}\sqrt{2\pi (a-b)}} = {e\over 2\pi}\cdot \left( 1-{b\over a}\right)^{-(a-b+1/2)}\left({a\over b}\right)^b \\
% <&{\exp\left(1+b+ {b\over 2a} + {b^2\over 2a^2}-{b^3\over a^2} \right) \over 2\pi}\left({a\over b}\right)^b < {\exp\left( 1+{b\over 2a} -{b^3\over 2a^2}\right) \over 2\pi}\left({ae\over b}\right)^b\\
% &<{ \exp(1.5) \over 2\pi}\left({ae\over b}\right)^b < \left({ae\over b}\right)^b.
% \end{align*}
Define 
\begin{equation}
\label{eq:deltas}
\delta_s := {\log((k-1)e^{k}) \over \log{1\over \alpha(k-1)}}
\end{equation} 

and set $t= (k-1-\delta_s)s$. Using $m< c^*_k n $  we obtain

\begin{align}\label{eq:expan}
{m\choose s}{n\choose t} \left({t\over n}\right)^{ks} <  &\left({nc^*_ke\over s}\right)^s  \left({ne\over t}\right)^t \cdot \left({t\over n}\right)^{ks}<  \left({nc^*_ke\over s}\right)^s  \left({ne\over t}\right)^t \cdot \left({t\over n}\right)^{ks} \nonumber \\
= &\left({nc^*_k\over s}\right)^s  \left({n\over t}\right)^{t-ks} e^{t+s}
= \left({nc^*_ke^{k-\delta_s}\over s}\right)^s \left({n\over (k-1-\delta_s)s}\right)^{-(1+\delta_s)s}\nonumber\\
<&  \left({nc^*_k\over s}\right)^s \left({n\over (k-1)s}\right)^{-(1+\delta_s)s} e^{ks} \nonumber \\
= & \left(\left({n\over (k-1)s}\right)^{-\delta_s } \cdot (k-1){e^{k} c^*_k}\right)^s .
\end{align}

Moreover from~\cite{fp12} we know that $c^*_k <1.$ Let $\beta$ be such that $(1+\beta)c^*_k =1$.
Substituting ${s\over n} \le \alpha $  in \eqref{eq:expan} and rewriting the terms in exponential form we obtain
\[ \left(e^{-\delta_s(\log{1\over\alpha (k-1)}) + \log ((k-1)e^k)}  c^*_k\right)^s = (1+\beta)^{-s}.\]
Therefore, for $\delta_s$ as defined in \eqref{eq:deltas} and $\alpha< {1/( k-1)}$, the probability that there exists a set  of edges of size $s$, where $n^{1/2}\le s\le \alpha n$, spanning at most $(k-1-\delta_s) s$ vertices is $O((1+\beta)^{-n^{1/2}}) =e^{-O(n^{1/2})}$. 

For $\log n \le s < n^{1/2}$, the corresponding probability is $O((1+\beta)^{-\log n})= o(1)$.
Now for $\alpha < {(e^{k}(k-1))^{-1\over k-2-\gamma} \over {k-1}}$ we obtain $\delta < k-2-\gamma$ as 
$$ \delta < {\log ((k-1)e^k) \over \log((k-1)e^k)^{{1\over k-2-\gamma}} } =k-2-\gamma,
$$
which completes the proof.
\end{proof}

\begin{proof}[Proof of Lemma~\ref{lem:expan}]
Recall that in the allocation graph $G$,  $R$ is the set of vertices furthest from the set of free vertices. The set of vertices at distance at most $s$ from $R$ is denoted by $N_{s}(R)$. Note that each occupied vertex in $G$ holds one item. By construction of the allocation graph $N_{s}(R)$ is the set of vertices representing the choices of items placed on vertices in $N_{s-1}(R)$. In the hypergraph setting where each item corresponds to an edge, $|N_{s}(R)|$ is the number of vertices spanned by the set of edges of size  $|N_{s-1}(R)|$. We now obtain the desired result by applying Lemma~\ref{lem:expanhyper}.

\end{proof}
% We define $\mu: =   {\log e^{k}(k-1)/ \log {(- \alpha(k-1))}} .$ For some fixed $\gamma >0$ we set 

The following corollary follows from the above two lemmas. 
\begin{corollary} \label{cor:maxLabel}
With high probability, the maximum label of any vertex in the allocation graph is $O(\log n)$.
\end{corollary}
\begin{proof}
Let $s$ be such that $N_{s}(R)\le \log n$ and $N_{s+1}(R)> \log n$. Clearly the distance of vertices in $N_s(R)$ from $R$ is atmost $\log n$. Let $d$ be the shortest distance of vertices in $N_{s+1}(R)$ to any set $S'\subset V$ such that $ |N_{d+\log n}(R) | \le \alpha n$ and $ |N_{d+\log n+1}(R) | > \alpha n$ . Then by Lemma~\ref{lem:expan},  we have with probability 
    \begin{align*}
   |N_{d+\log n }(R) | > (1+\gamma)^d  |N_{s+1}(R)|, 
  \end{align*}
which implies that 
 $d<  \log_{1+\gamma} {\alpha n \over \log n} $ with high probability.

Note that the shortest distance of vertices in $V\setminus S'$ to $F$ is a constant $C(\alpha,\delta)$ for $\delta$ defined in Lemma~\ref{lem:dist}. 
Moreover, by Proposition~\ref{prop:lev} the label of any vertex is upper bounded by its distance to the set of free vertices, which by above arguments is atmost $d+\log n +1 + C(\alpha,\delta)$. Therefore, the label of any vertex $v$ is such that $L(v) = O( \log  n)$.
\end{proof}
We now prove our main theorem.
%\vspace{-4pt}
\begin{proof}[Proof of Theorem~\ref{thm:main}]
Set $\alpha$ as in Lemma \ref{lem:expanhyper}. Then by Lemma~\ref{lem:dist}, with probability $1-O(1/n)$, there exists a $C=C(\alpha, \varepsilon)$ and a set $S$  such that $|S| \ge (1-\alpha)n$ and every vertex $v\in S$ satisfies $d(v,F)\le C.$
Let $T+1$ be the maximum of the distances of vertices in $R$ to $S$, i.e.,
\[ T = \max_{v \in R} d(v,S) -1 .\]
Clearly the number of vertices at distance at most $T$ from $R$ is at most $\alpha n$, i.e., $|N_T(R)|\le \alpha n$. 
Moreover for all $t<T$, $|N_t(R)| < |N_T(R)|$. 
% Then by Lemma~\ref{lem:expan}, for all $t\le T$ the following holds with high probability,
% \[ |N_{t+1}(R)| >  \left(1+\gamma \right)| N_t(R)|.\]
% One can check that for $\gamma >0$ and $\alpha$ as chosen above, $ \delta < k-2-\gamma.$
The total distance of all vertices from $F$ is then given by
 \[D=  \sum_{v\in N_{T}(R)} d(v,F) + \sum_{v\in S} d(v,F). \]
 As every vertex in $S$ is at a constant distance from $F$,  we obtain $\sum_{v\in S} d(v,F)=O(n)$ with probability $1-O(1/n)$.
 Note that for every $i>0$, $ | N_{i}(R)| - | N_{i-1}(R)|$ is the number of vertices at distance $i$ from $R$.
Therefore,
 \begin{align*}
 &\sum_{v\in N_{T}(R)} d(v,F)=  (T+C)|N_{0}(R)| + \sum _{i=1}^{T} (T+C-i)(| N_{i}(R)| - | N_{i-1}(R)|)\\
 &=(T+C)|N_{0}(R)| + \sum _{i=1}^{T} (T-i)(| N_{i}(R)| - | N_{i-1}(R)|) +C \sum _{i=1}^{T}( | N_{i}(R)| - | N_{i-1}(R)|)\\
 &=(T+C)|N_{0}(R)| + \sum _{i=1}^{T} (T-i)(| N_{i}(R)| - | N_{i-1}(R)|) +C ( | N_{T}(R)| - | N_{0}(R)|)\\
 &= \sum _{i=1}^{T} \bigg((T-i)(| N_{i}(R)| - | N_{i-1}(R)|) +|N_0(R)|\bigg)+C \cdot | N_{T}(R)| =\sum _{i=0}^{T-1} |N_i(R)| + O(n).
 \end{align*} 
 To bound the above sum we we observe that for $i$ such that $|N_i(R))|< n^{1/2}$ combining with the fact that for any $i$, $|N_i(R)|< |N_{i-1}(R)|$  the following holds
$$\sum_{i}|N_i(R)|< \sum_{j=1}^{n^{1/2}}j = n^{1/2}\cdot {(n^{1/2} +1) \over 2}= O(n)$$
with probability $1$. 
For all other $i$ such that $n^{1/2}\le|N_i(R)|\le |N_{T}(R)|$  by Lemma~\ref{lem:expan} ,  following holds with probability $1-e^{-O(n^{0.5})}$,
\[  | N_i(R)|<{|N_{i+1}(R)| \over \left(1+\gamma \right)}.\]
Therefore for such $i$, we obtain
% with probability $1-o(n^{0.5})$, we have $|N_{T-j}(R) | <    {| N_{T}(R)| \over (k-1-\delta)^j } .
% $ 

\[
\sum _i| N_i(R)) |  <  |N_{T}(R)|  \sum_{j}   {1 \over (1+\gamma)^j } =O(n),\]
with probability $1-e^{-O(n^{0.5})}$. We can therefore conclude that $D$ (which is an upper bound for the run time of LSA) is upper bounded by $O(n)$ with probability $1-O(n^{-1})$, thereby completing the first part of the proof.

To bound the expected run time, first note that
$$ E \left(\sum_{v\in S} d(v,F)\right) = n (1-O(1/n) + n^2\cdot (1-O(1/n) = O(n) ,$$
as in the worst case the sum of all the labels can be atmost $n^2$ (see discussion after Lemma~\ref{lem:graph}).
We now bound the expected sum of vertex labels of vertices in $V\setminus S$. Note that for $i$ such that $| N_i(R)) |\le n^{1/2}$ we bounded the sum by $n$ with probability $1$. 

For all other $i$ the sum is bounded by $O(n)$ with probability at least $1-e^{-O(n^{0.5})}$. This implies that for such $i$ 
$$E\left(\sum_{i}|N_i(R)|\right) < n (1-o(1)) + n^2 e^{-O(n^{0.5})} = O(n).$$
We note that the above bound on expected run time of LSA holds in all cases whether an allocation exists or not.
\end{proof}
We obtain the following corollary about maximum matchings in left regular random bipartite graphs. Recall that a bipartite graph $G=(L\cup R;E)$ is  $k$-left regular if each vertex $v\in L$ has exactly $k$ neighbors in $R$.

\begin{corollary}
For $k\ge 3$ and $c^*_{k}$ as defined in Theorem~\ref{thm:mainO}, let $G=(L\cup R;E)$ be a random $k$-left regular bipartite graph such that ${|L| / |R| }<c^*_{k}$.  The local search allocation method obtains a maximum cardinality matching in $G$ in time $O(|R|)$ with probability $1-o(1)$.
\end{corollary}
\begin{proof}
We assign label $0$ to each of the vertices in $R$ initially. Each vertex in $L$ can be considered as an item and let $R$ be the set of locations. The $k$ random choices for $v\in L$ (item) are the $k$ random neighbors of $v$. We can now find a matching for each $v\in L$ by using Algorithm~\ref{algo:orientEdge}.
\end{proof}

\section{Experiments}

In this section we discuss the performance of our proposed LSA algorithm on randomly generated instances with density less than the threshold and then on real-world large datasets with arbitrary densities.
 The rationale of our evaluation is two-fold. First, we establish the effectiveness of LSA for randomly generated instances with densities close to the threshold in terms of abstract cost measures and compare it with the state of the art method employed for Cuckoo Hashing for a large number of randomly generated instances. Second, we would want to validate the performance of LSA in terms of wall-clock times on large real-world bipartite graphs with arbitrary densities and structure ( i.e. these are not necessarily left regular bipartite graphs). 
 
\subsection{Performance on Random Graphs}

We present some simulations to compare the performance of local search allocation with the random walk method which (to the best of our knowledge) is currently the state-of-art method and so far considered to be the fastest algorithm for the case $k\ge 3$. We recall that in the random walk method we choose a location at random from among the $k$ possible locations to place the item. If the location is not free, the previous item is moved out.  The moved out item again chooses a random location from among its choices and the procedure goes on till an empty location is found. In our experiments we consider $n\in [10^5, 5\times 10^6]$ locations and $\lfloor cn \rfloor$ items. The $k$ random locations are chosen when the item appears. All random numbers in our simulations are generated by \emph{MT19937} generator of GNU Scientific Library~\cite{gnu}. 

Recall that a move is either placing an item at a free location or replacing it with other item. 
In Figure~\ref{fig:1} we give a comparison of the total number of moves (averaged over $100$ random instances) performed by local search and random walk methods for $k=3$ and $k=4$. Figure~\ref{fig:2} compares the maximum number of moves (averaged over $100$ random instances) for a single insertion performed by local search and random walk methods. Figure~\ref{fig:3} shows a comparison when the number of items are fixed and density (ratio of number of items to that of locations) approaches the threshold density. Note that the time required to obtain an allocation by random walk or local search methods is directly proportional to the number of moves performed. 
\begin{figure*}[h!]
   \centering  
     \subfigure[$k=3$, $c=0.90 ~(c^*_3\approx 0.917)$]{\label{fig:total}\includegraphics[width=0.45\textwidth]{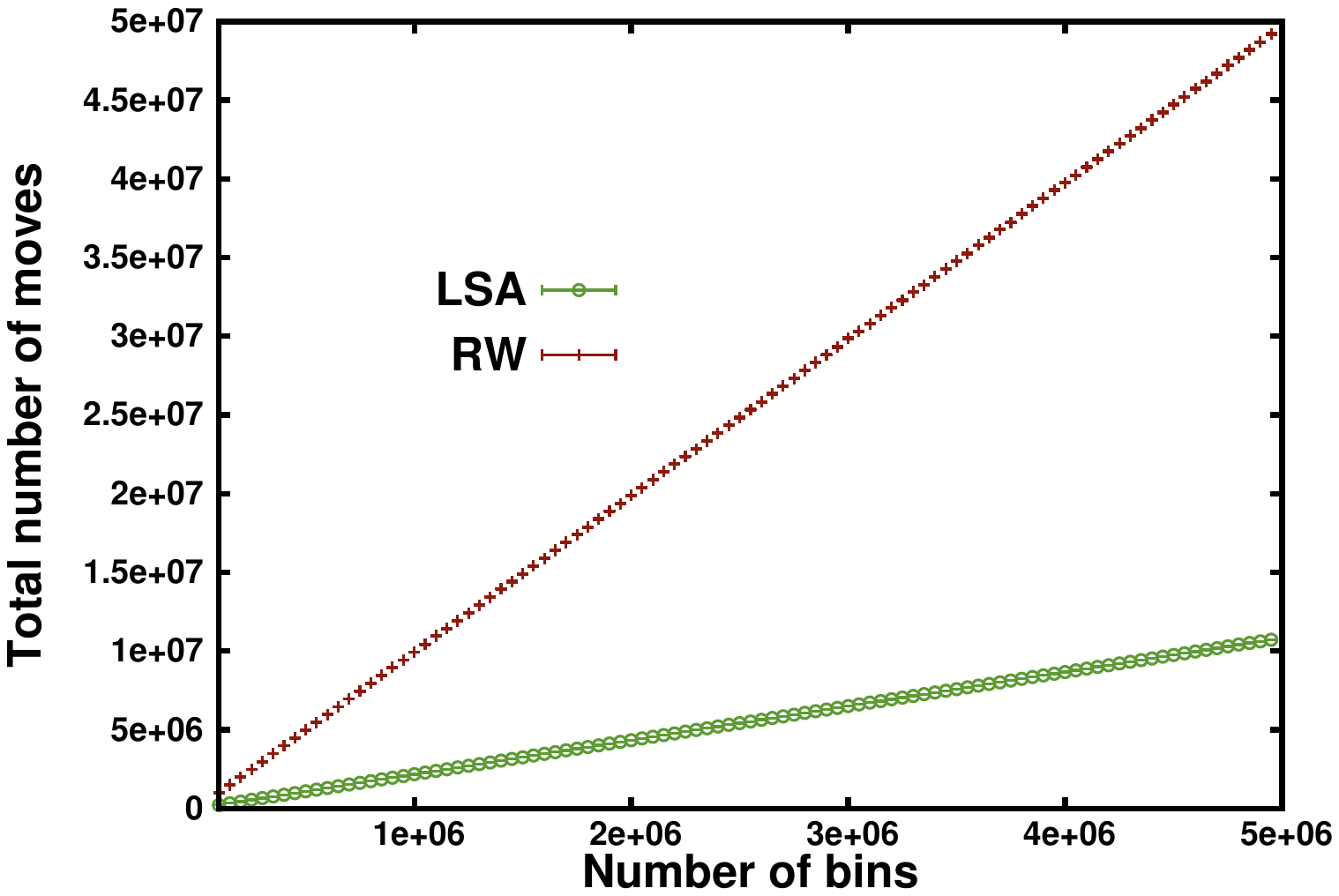}}
   \quad
     \subfigure[$k=4$, $c=0.97 ~(c^*_4\approx 0.976)$] {\label{fig:all-3}\includegraphics[width=0.45\textwidth]{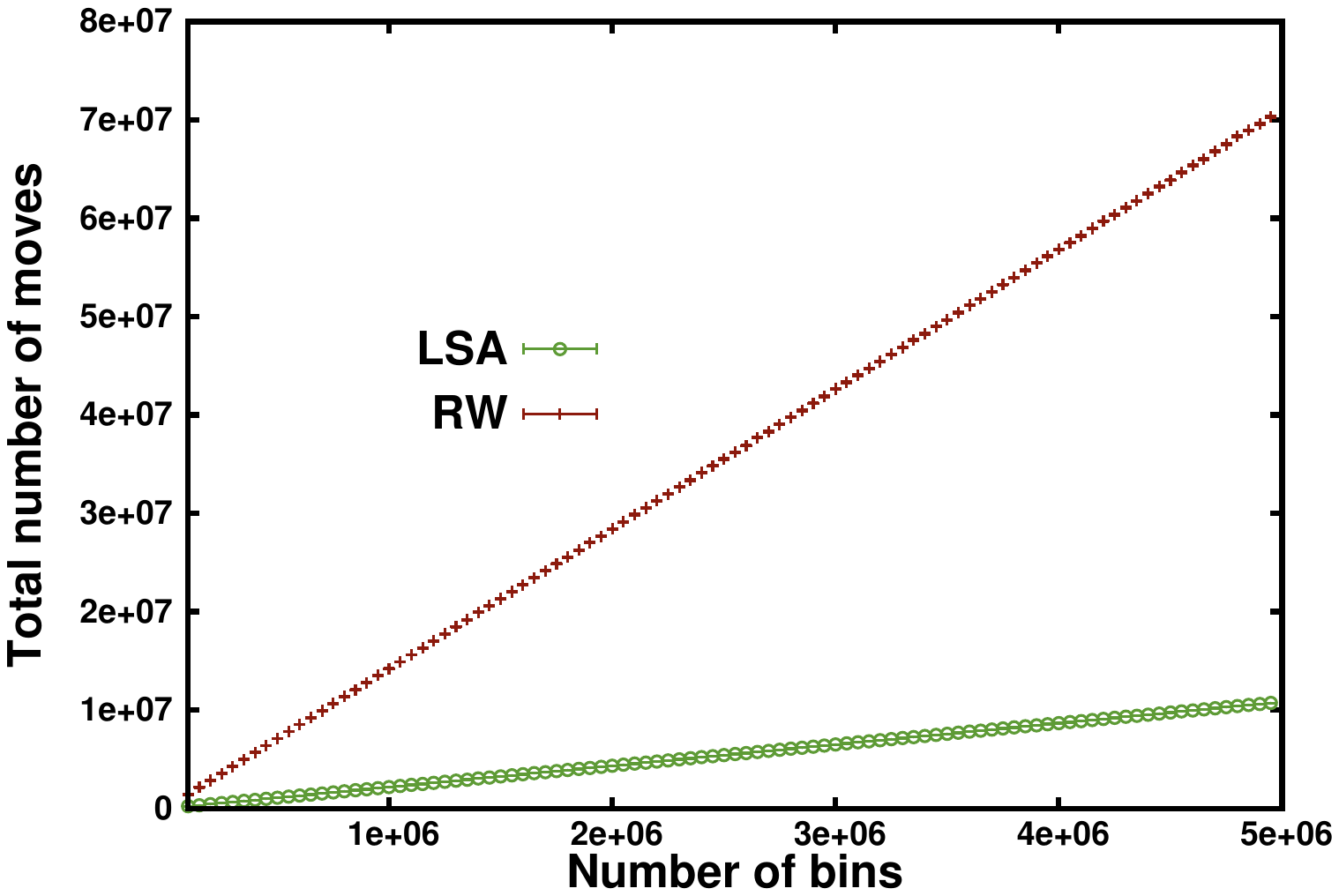}}
     \vspace{-8pt}
   \caption{Comparison of total number of moves performed by local search and random walk methods.}
   \label{fig:1}
\end{figure*}
\begin{figure*}[h!]  
   \centering  
    \subfigure[ $k=3$, $c=0.90 ~(c^*_3\approx 0.917)$.]{\label{fig:max}\includegraphics[width=0.45\textwidth]{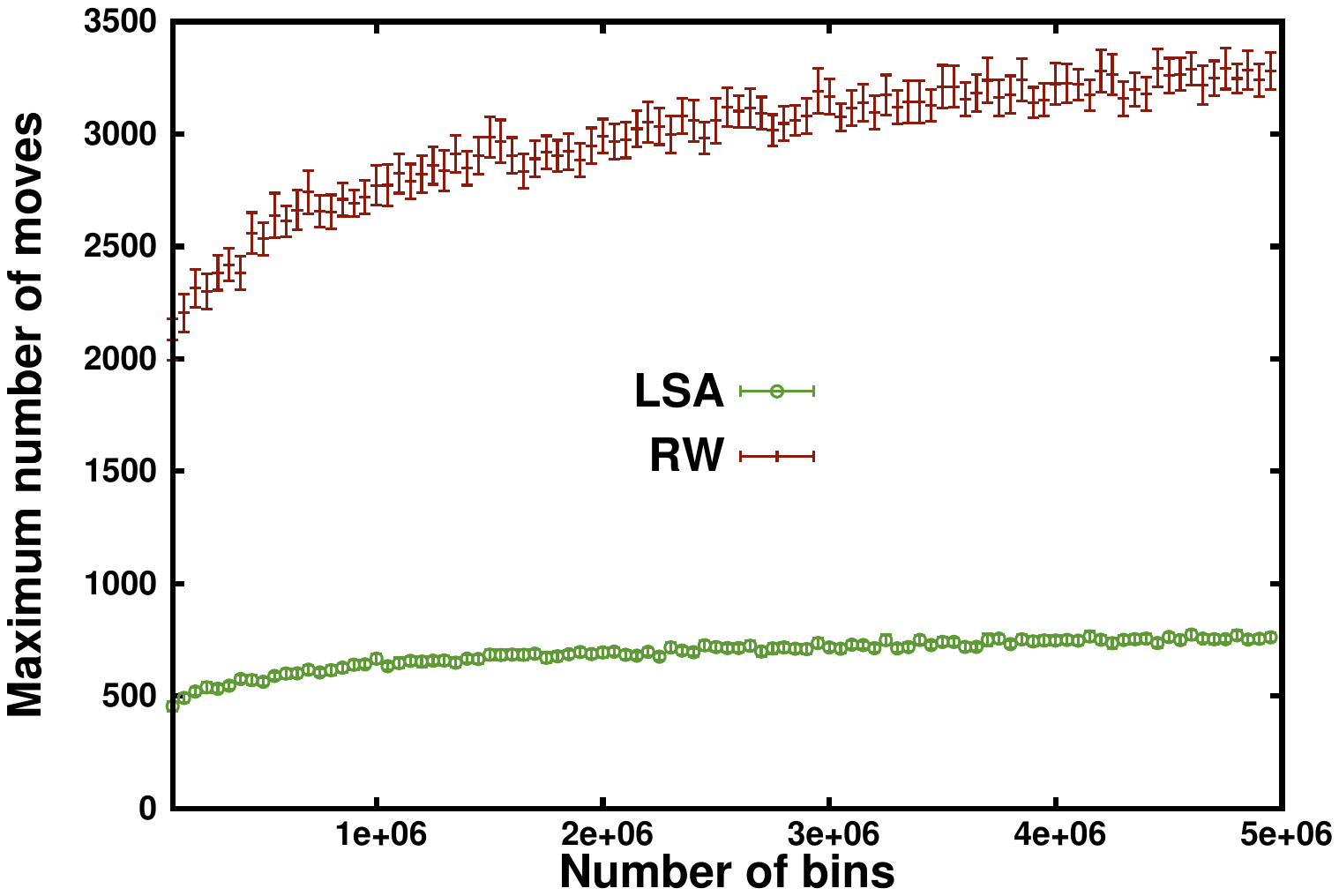}}
   \quad
 \subfigure[$k=4$, $c=0.97 ~(c^*_4\approx 0.976)$.]{\label{fig:his to-3}\includegraphics[width=0.45\textwidth]{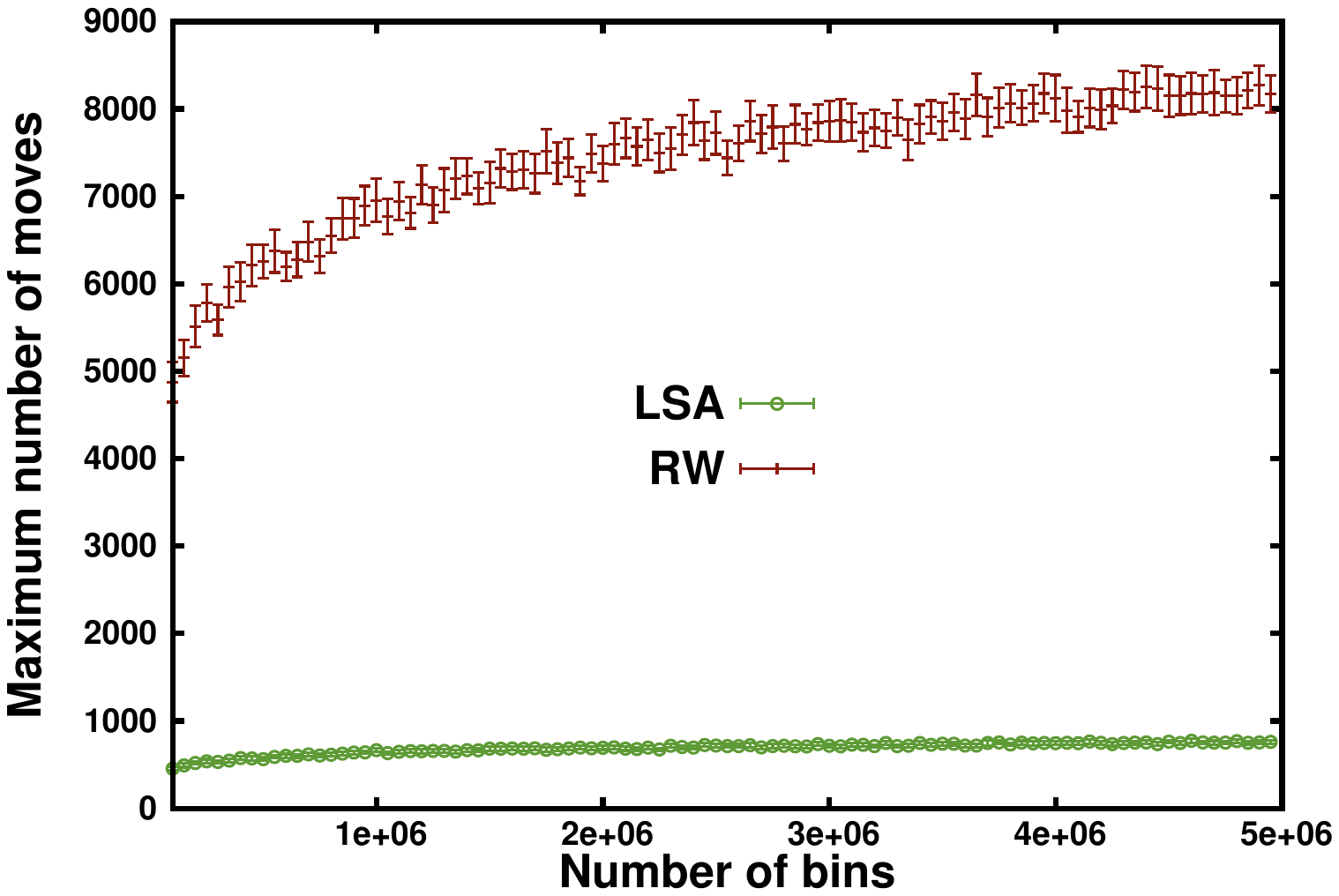}}
   \vspace{-8pt}
   \caption{Comparison of maximum number of moves performed by local search and random walk methods}
    \label{fig:2}
\end{figure*}
\begin{figure*}[h!]
   \centering  
     \subfigure[$k=3$, $c\le 0.915 ~(c^*_3\approx 0.917)$]{\label{fig:fix-3}\includegraphics[width=0.45\textwidth]{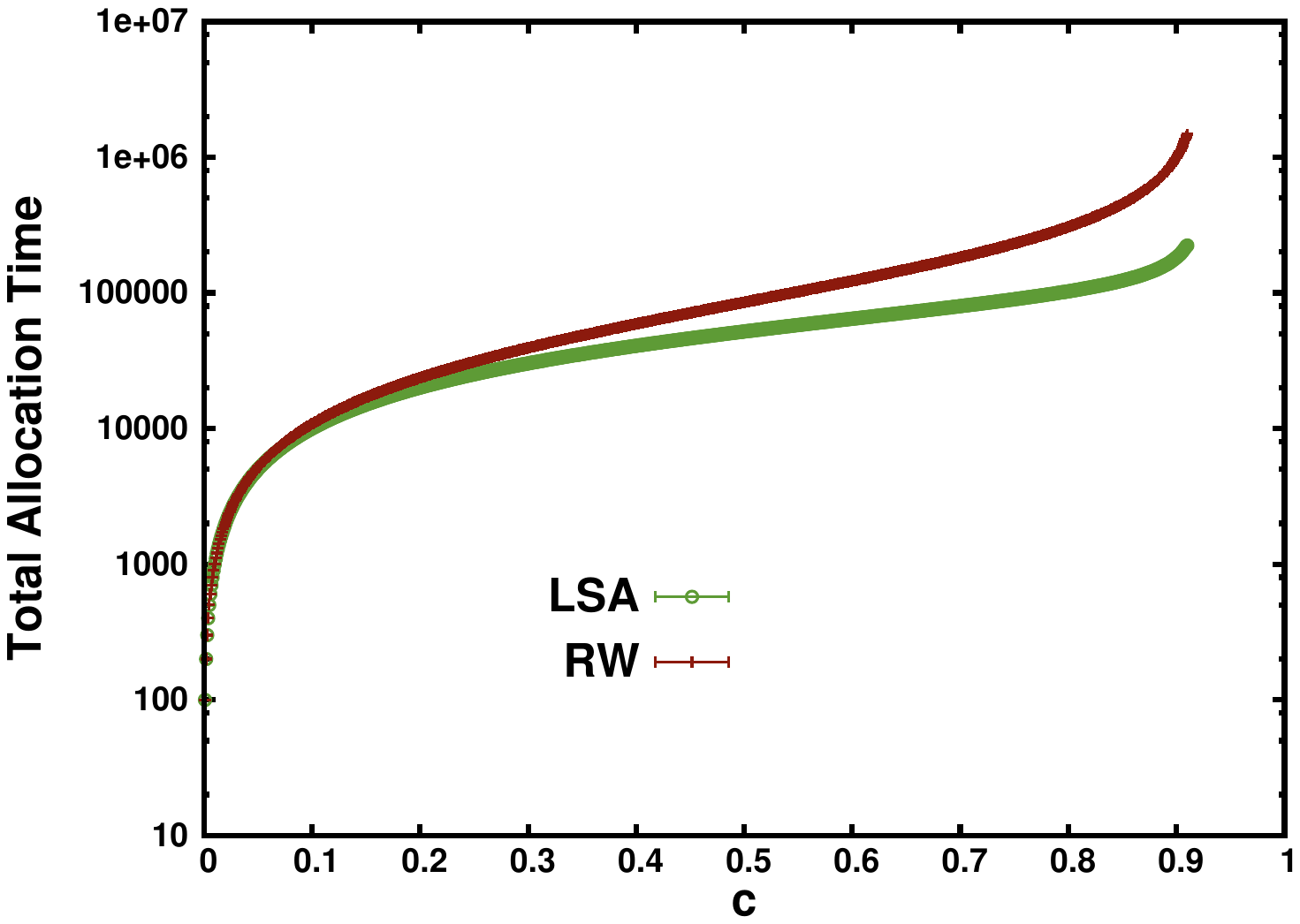}}
   \quad
       \subfigure[$k=3$, $c\le0.915 ~(c^*_3\approx 0.917)$] {\label{fig:fix-3-total}\includegraphics[width=0.45\textwidth]{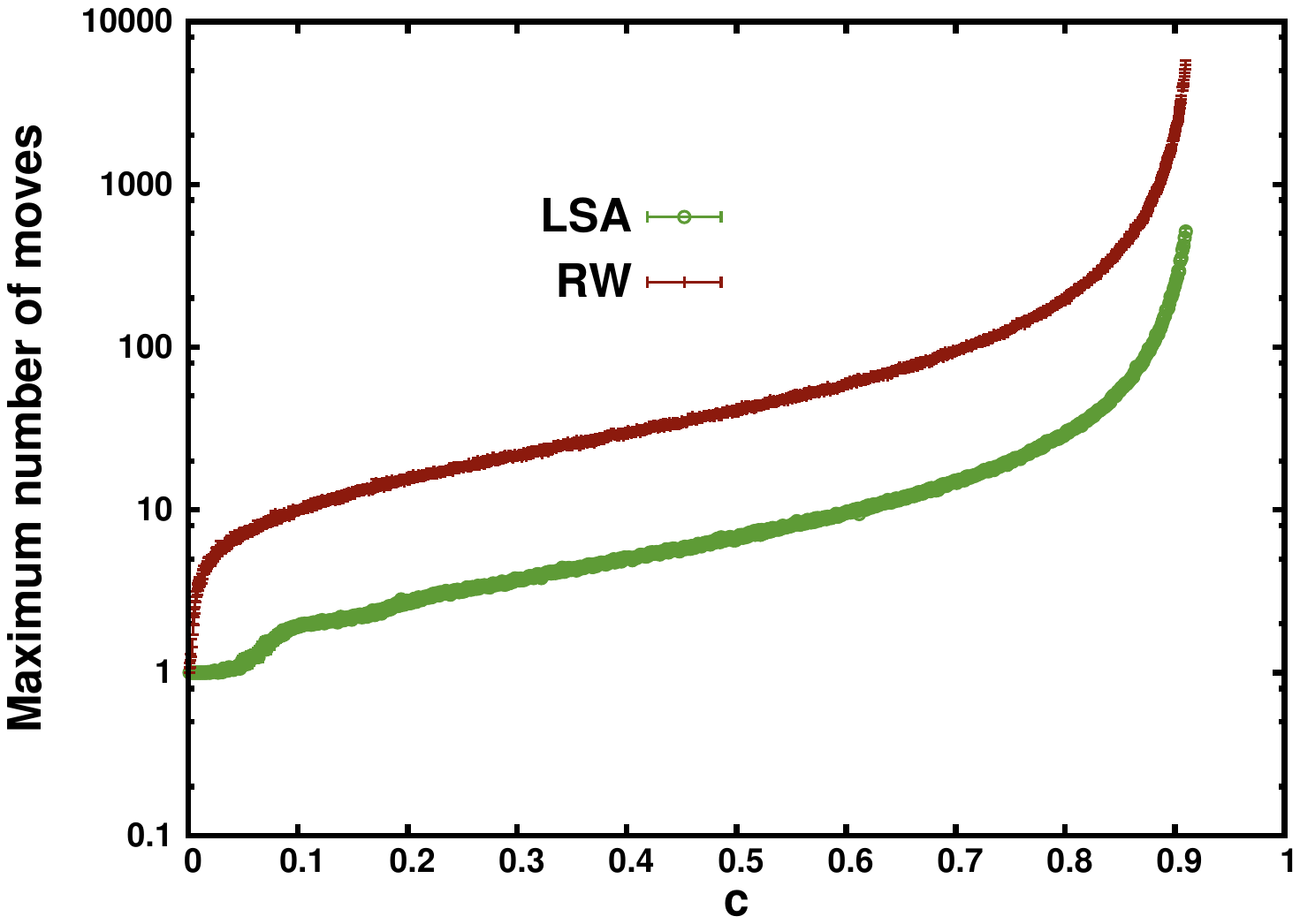}}
       \vspace{-8pt}
   \caption{ Comparison of total number of moves and maximum number of moves (for fixed number of locations, $n=10^5$) performed by local search and random walk methods when density c approaches $c^*_k$.}
      \vspace{-1pt}
   \label{fig:3}
\end{figure*}
\begin{figure}[h!]
   \centering  
    \includegraphics[width=0.85\textwidth]{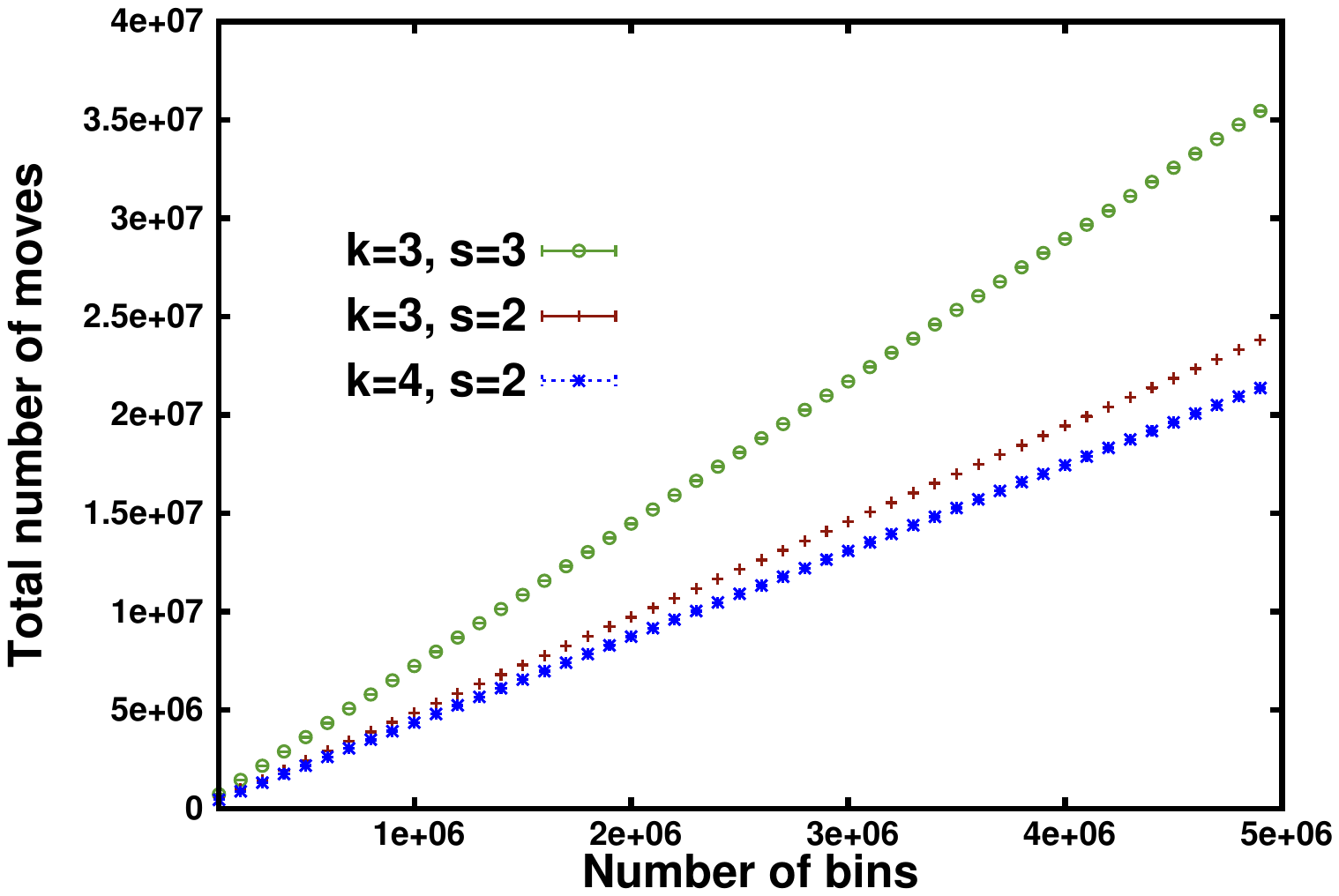}
    \caption{Total number of moves for the case where bin capacities (maximum load, $s$) is greater than 1.}
    \label{fig:totalGen}
    \end{figure}

\begin{algorithm}[h!]
\caption{AssignItem ($x, \mathbf{L},\mathbf{T}$)}
\label{algo:orientEdgegen}
\begin{algorithmic}[1]
\STATE Choose an item $v$ among the $k$ choices of $x$ with minimum label $L(v)$.
\IF{$(L(v)>=n-1)$ }
\STATE $\mathbf{EXIT}$  ~~~~~~~~~~~~~~~~~~~~~~~ $\rhd${\textbf{Allocation does not exist}}
\ELSE
\IF{$({\textsc{Items}}(v)> s-1)$ }
\STATE $L(v) \leftarrow 1+ \min{(L(u)| u \neq v \text{~and $u \in x$})}$
\ENDIF
\IF{$(\textsc{Items}(v)==s )$}
\STATE Choose an item (call it b) randomly from the $s$ items in v
\STATE $y\leftarrow b$~~~~~~~~~~~~~~~~~~ $\rhd${\textbf{Move that replaces an item}}
\STATE Place $x$ in $v$
\STATE $\mathbf{CALL}$ {AssignItem($y, \mathbf{L},\mathbf{T}$)}
\ELSE  
\STATE Place $x$ in $v$ ~~~~~~~~~~~~~~~~~~ $\rhd${\textbf{Move that places an item}}
\ENDIF
\ENDIF
\end{algorithmic}
\end{algorithm}
We also consider the case when each location can hold more than one item. To adapt LSA for this setting we make a small change, i.e., the label of a vertex (location) stays 0 until it is fully filled. Algorithm~\ref{algo:orientEdgegen} gives the modified procedure for the general location capacities. Here {\sc{Items}}$(v)$ gives the number of items already placed in $v$. Let the location capacity or maximum load allowed be $s$. Figure~\ref{fig:totalGen} suggests that the total number of moves are linear in the number of locations for the cases $k=3,4$ where the maximum location capacity is greater than $1$.

We remark that local search allocation has some additional cost, i.e., the extra space required to store the labels. Though this space is $O(n)$, local search allocation is still useful for the applications where the size of objects (representing the items) to be allocated is much larger than the labels which are integers. Moreover, with high probability, the maximum label of any vertex is $O(\log n)$. Many integer compression methods~\cite{inp:sgl10} have been proposed for compressing small integers and can be potentially useful in our setting for further optimizations. Also in most of the load balancing problems, the speed of finding an assignment is a much desired and the most important requirement.

\begin{table*}[ht!]
\centering
\footnotesize
\begin{tabular}{@{}l l l l @{}}
\toprule
%\hline \noalign{\smallskip}
\multicolumn{1}{l}{} & Maximum Number of Moves                  & Wall-clock times     & Result Size  \\
\hline \noalign{\smallskip}
LSA        &   1      & 12    & 1,029,449  \\
       &   2      & 12    & 1,080,006  \\
        &   4      & 12    & 1,082,199  \\
        &   5      & 16    & 1,082,214  \\
        &   10      & 15    & 1,082,214  \\
        &   50     & 15    & 1,082,214  \\
        &   100     & 15    & 1,082,214  \\
        &   1000     & 15    & 1,082,214  \\
        &   10,000    & 27    & 1,082,214  \\
        &   100,000    & 136    & 1,082,214  \\
        &   $n$    & 1,887    & 1,082,214  \\
\midrule

Hopcroft-Karp                &  & 12,605 & 1,082,214 \\
%\hline \noalign{\smallskip}

\hline \noalign{\smallskip}
\end{tabular}
\caption{Performance of LSA on \textsf{Delicious} dataset. Time is measured in seconds.}
\label{table:softlabel}
\end{table*}

\subsection{Performance on Real-world graphs}

Next, we compare our runtime performance to the optimal algorithm proposed by Hopcroft et al.~\cite{hopcroft1973n}. In this experiment we want to study the effect of number of allowable moves on (a) the actual wall-clock times , (b) the result quality in terms of the size, or number of edges, of the final matching produced (refer Figure~\ref{table:softlabel}). We selected the following representative realworld dataset for our experiments:

\begin{itemize}
  \item \textsf{Delicious dataset :} The \textsf{Delicious} dataset spans nine years from 2003 to 2011 and contain about 340 mio. bookmarks, 119 mio. unique URLs, 15 mio. tags and 2 mio. users~\cite{zubiaga2013harnessing}. Each bookmarked URL is time stamped and tagged with word descriptors. The nodes in one of the sets are URLs and in the other are its corresponding bookmarks.

  % \item \textsf{Yahoo Ad dataset : } Computational advertising used bi-partite graph matching algorithms for its ad placement decisions. 
\end{itemize}

We first observe that the optimal result in the \textsf{Delicious} dataset, i.e. 1,082,214, is already obtained when the limit on the allowable moves is only 5. We are of course sure about the optimality of the procedure when the maximum allowable moves is set to $n$ and that already is 10x improvement over the time taken by Hopcroft-Karp algorithm. For lower allowable limits of 5 and 10 the performance improvements are almost 1000x. Interestingly, as we increase the limit on the allowable moves to place any item (match any edge), the runtime does not change showing that only a small of defections are sufficient to arrive at an optimal result. However, at higher limits, indeed other permutations are explored (in this case unsuccesfully) resulting in increased runtimes. The stopping creteria unlinke in case of perfect matchings cannot be predetermined in general. In future we plan to devise methodology to stop the algorithm when the maximum matching is retrieved. In any case when the limit is set to $n$, that would guarantee optimality, we still perform an order of magnitude faster than the optimal algorithm of Hopcroft-Karp.

\section{Conclusions and Outlook}
\label{chap:conc}

In this article, we proposed and analysed an insertion algorithm, the \emph{Local Search Allocation} algorithm, for cuckoo hashing that runs in linear time with high probability and in expectation. Our algorithm, unlike existing random walk based insertion methods, always terminates and finds an assignment (with probability 1) whenever it exists. We also obtained a linear time algorithm for finding perfect matchings in general large bipartite graphs.

We conducted extensive experiments to validate our theoretical findings and report an order of magnitude improvement in the number of moves required for allocations as compared to the random walk based insertion approach. Secondly, we considered a real world social bookmarking graph dataset to evaluate the performance of our bipartite graph matching algorithm. We observe an order of magnitude improvement when the maximum allowable number of moves is set to $n$, but more interestingly we observe that the optimal solution is already reached at a small allowable limit of $5$ with a substantial performance improvement of almost three orders of magnitude over Hopcroft-Karp algorithm.

It should be noted that although the space complexity for label maintenance is $O(n)$, the number of bits required to encode each label is logarithmic in the maximum allowable moves. This allows compact representations of these labels in memory even without using integer encoding schemes that might further improve memory footprints while storing small integer ranges.

In the future we would like to consider other generalized variants of graph matching problems using such a label propagtion scheme. Also interesting to investigate is the impact of graph properties like diameter, clustering coefficients etc. on the only parameter in our algorithm, i.e., maximum allowable moves. This would go a long way in automatic parameterization of LSA.

\bibliographystyle{plainnat}
\bibliography{algorithmica15}
\end{document}